\documentclass[%
 aps,pra,showkeys,
 amsmath,amssymb,
preprint,%
]{revtex4-1}

\usepackage{wrapfig}
\usepackage{graphicx}
\usepackage{bm}
\usepackage{amssymb}
\usepackage{amsmath}
\usepackage{amsthm}
\usepackage{color}
\usepackage{float}

\font\bb=msbm10 scaled 1200

\newcommand{\beq}{\begin{equation}}
\newcommand{\eeq}{\end{equation}}
\newcommand{\bea}{\begin{align}}
\newcommand{\eea}{\end{align}}

\newcommand{\fl}{\raggedright}

\newcommand{\nonu}{\nonumber}

\newcommand{\rf}[1]{(\ref{#1})}

\newcommand{\pa}{\partial}

\renewcommand{\bm}[1]{\mbox{\boldmath${#1}$}}
\newcommand{\M}{\bm M}
\newcommand{\R}{\mbox {\bb R}}

\renewcommand{\L}{\mbox{$\cal{L}$}}

\newcommand{\Tr}{\mbox{\sf tr\,}}

\newcommand{\sst}{\scriptstyle}
\newcommand{\dst}{\displaystyle}

\newcommand{\ipd}{p}

\newcommand{\pv}{-\hspace{-12pt}\int}
\newcommand{\oh}{\frac{1}{2}}

\renewcommand{\phi}{\varphi}
\renewcommand{\rho}{\varrho}
\renewcommand{\theta}{\vartheta}

\renewcommand{\leq}{\leqslant}
\renewcommand{\geq}{\geqslant}
\def\res{\mathop{\textrm{Res}}\limits} 

\newtheorem{theorem}{Theorem}

\begin{document}

\title{Lyapunov exponent of the random Schr\"{o}dinger operator \\ with short-range correlated noise potential}

\author{Yuri~A. Godin}
\email{ygodin@uncc.edu}
\affiliation{%
Department of Mathematics and Statistics, University of North Carolina at Charlotte, Charlotte, NC 28223, U.S.A.
}%
\author{Stanislav Molchanov}
  \email{smolchan@uncc.edu}
\affiliation{%
Department of Mathematics and Statistics, University of North Carolina at Charlotte, Charlotte, NC 28223, U.S.A.
}%

\author{Boris Vainberg}
\email{brvainbe@uncc.edu}
\affiliation{%
Department of Mathematics and Statistics, University of North Carolina at Charlotte, Charlotte, NC 28223, U.S.A.
}%

\date{\today}

\begin{abstract}
We study the influence of disorder on propagation of waves in one-dimensional structures. Transmission properties of the process governed by the Schr\"{o}dinger equation with the white noise potential can be expressed through the Lyapunov exponent $\gamma$ which
we determine explicitly as a function of the noise intensity $\sigma$ and the frequency $\omega$. We find uniform two-parameter asymptotic expressions for $\gamma$ which allow us to evaluate $\gamma$ for different relations between $\sigma$ and $\omega$. The value of the Lyapunov exponent is also obtained in the case of a short-range correlated noise, which is shown to be less than its white noise counterpart. 
\end{abstract}


\maketitle

{\hfill {\em In memory of M. I. Vishik, great mathematician and dear colleague.}}

\section{Introduction}
\setcounter{equation}{0}

Analysis of the Lyapunov exponent plays an important role in the problems of
 propagation of waves.  Much attention in recent
years has been given to periodic and disordered one-dimensional 
structures (see, e.g., \cite{MS08} and references therein). Application of such 
structures include photonic crystals, coupled-resonator optical waveguides,
slow-light devices, metamaterials, semiconductor superlattices with given transmission
properties and more. Characteristics of real materials, however,
may be considerably different from predicted ones due to inevitable
disorder of geometric and material parameters. Structural perturbation and manufacturing inaccuracies may
lead to a substantial reduction of performance \cite{F09}, \cite{F08} and
even to complete suppression of wave propagation \cite{T07}, \cite{M08}.
In the strongest form the disorder may manifest itself in the Anderson type localization,
which is a general wave phenomenon for classical and quantum waves in disordered systems.
Due to multiple scattering, destructive interference
occurs, so that the wave amplitude decreases exponentially in the medium. The
rate of decay of the solution is defined by the Lyapunov exponent, $\gamma$ (the inverse of 
the localization length). 

The Lyapunov exponent is of major importance in the problems of
slowing down light by orders of magnitude, a phenomenon which has been extensively discussed in the literature \cite{FV06}.
Such a possibility can be useful in a variety of optical and microwave applications.
Choosing wave frequency near stationary points of the electromagnetic dispersion relation
allows one to reduce the group velocity by two orders of magnitude \cite{FV06}-\cite{NYSTTY01}. 
Another model for slowing down light proposed in \cite{MV05} is based on the periodic necklace-type waveguide.
An essential common feature of such devices is that the frequency of propagating wave must be located near the spectral band edge or degenerate band edge of the corresponding medium. 

One of the major factor, however, limiting performance of optical waveguides is their structural imperfection, and  
investigation of the influence of different kinds of disorder on the Lyapunov exponent becomes of primary importance in manufacturing of optical devices.
Papers \cite{JT08a}-\cite{JT08} performed sensitivity analysis of photonic crystals near degenerate and regular spectral band edge under perturbations of geometrical and material parameters, and the impact of these perturbations was compared. %photonic crystal responses. 
Numerical simulations were used to investigate the sensitivity of field enhancement effects at Fabry-Perot resonances against the bandwidth of the excitation. The gigantic amplitude increase produced by such resonances was observed, which dramatically deteriorates under small perturbations of the thickness of layers. In particular, the peak field intensity of a 16 layer stack regular band edge photonic crystal decreases by an order of magnitude under $1\%$ perturbation of the layer thicknesses.

The central idea of chaotic behaviour of dynamical systems is 
widely seen as strong dependence on initial conditions \cite{ABC04}.
In a chaotic system, a very small change in the initial conditions can cause a big change in the trajectory. Quantitative criterion for sensitivity is characterized by the Lyapunov exponent (also called the maximal Lyapunov exponent) which measures exponential growth rate of the solution.

In all the applications, the influence of the strength of perturbation $\sigma$ is essentially different depending on how close the frequency $\omega$ is to the band edge.
While $\gamma \sim \sigma^2$ in the bulk of the band, it has order $\gamma \sim \sigma^{\frac{2}{3}}$ near the band edges \cite{DG84}-\cite{GMV11}. Therefore, it is important to construct uniform two-parameter asymptotic expansion of the Lyapunov exponent for all frequencies of wave propagation. This construction is a central goal of the present paper.

Let $\gamma_1=\gamma_1(\omega_1,\sigma)$ be the Lyapunov exponent of the optical equation
\begin{equation}
 -\psi^{\prime \prime}(x) = \omega^2_1 n^2(x,\sigma) \psi(x),
 \label{oe}
\end{equation}
where $0<\sigma \leq \sigma_0$ denotes the strength of perturbation of either the value of the periodic refraction index $n(x,0)$
or the length of its period, and $\omega_1$ is located in an $\Omega$-neighbourhood of a non-degenerate band edge.
Let also $\gamma=\gamma(\omega,\sigma)$ be the Lyapunov exponent of the Schr\"{o}dinger equation
\begin{equation}
 -\psi^{\prime \prime}(x) + \sigma \dot{w}(x)\psi(x) = \omega^2  \psi(x),
 \label{sche}
\end{equation}
with the white noise potential $\dot{w}(x)$ and the frequency $\omega$ located near
the spectral edge $\omega=0$. The frequency $\omega$ is defined by $\omega_1$ using the relation 
\begin{equation}
 \cos \omega = \frac{1}{2}\, \Tr \M (\omega_1), \quad \omega_1 \in \Omega.
\end{equation}
Here $\M$ is the transfer matrix of \rf{oe}.
Then in \cite{GMV11} it was shown that $\gamma_1$ and
$\gamma$ estimate each other uniformly. That means 
\begin{equation}
 c_1 \gamma \leq \gamma_1 \leq c_2 \gamma,
\end{equation}
where positive constants $c_1$ and $c_2$ depend only on $\sigma_0$ and $\Omega$ but not on the values of $\sigma$ and $\omega_1$. 
Thus, a uniform two-parameter asymptotic expansion of the Lyapunov exponent of the Schr\"{o}dinger equation \rf{sche} provides a uniform estimate of the Lyapunov exponent of the wave equation \rf{oe}.
Because of this, we consider in details the Schr\"{o}dinger equation
with different relations between the intensity of noise $\sigma$ and the frequency $\omega$.
We will obtain a uniform asymptotic expansion of $\gamma$ followed by numerical evaluation of $\gamma$ to determine the range of parameters for which the asymptotic formulas can be used.

There are different interpretations of the Schr\"{o}dinger equation with white noise potential.
The Lyapunov exponent defined by the equation with white noise potential is usually understood as the limit of ones defined by the equation with regular potentials whose correlation length $\Delta$ tends to zero (cf. \cite{LGP88}).
In our approach we pass to the limit in the equation itself that results in a system of It\^{o} stochastic  differential equations. Thus, we define the Lyapunov exponent through solutions of the corresponding $\Delta$-independent stochastic equations.

Papers \cite{PW92}-\cite{PW91} used different definition of equation with white noise potential. We obtain the same result for $\gamma(\omega)$ as the one obtained in \cite{PW92}-\cite{PW91} for the leading term of the asymptotics of $\gamma_\Delta (\omega)$ as $\Delta \to 0$. Moreover, our approach (which is also convenient for numerical simulations) allows us
to find the second term $c_1$  of $\gamma_\Delta (\omega)=\gamma (\omega)-c_1 \Delta + \ldots$, where $c_1>0$, which provides the rate of decay of the Lyapunov exponent when $\omega$ is fixed and the correlation length $\Delta$ increases.

\section{Integral representation of the Lyapunov exponent}
\setcounter{equation}{0}

In this paper we study the asymptotic behaviour of the Lyapunov exponent
$\gamma(\lambda,\sigma)$ of the random one-dimensional Schr\"{o}dinger
operators
\begin{equation}
\begin{array}{l}
 H\psi=-\psi^{\prime \prime}(x) + \sigma \dot{w}(x) \psi = \lambda \psi,
\quad x >0, \\[2mm]
\psi(0) = \cos \theta_0, \quad \dst \frac{\psi^\prime (0)}{\omega} = \sin
\theta_0,
\end{array}
\label{H}
\end{equation}
as a function of the energy $\lambda$  and the coupling constant (intensity)
$\sigma$ of the white noise potential.
Here $\omega=\sqrt{|\lambda|}$ (equations with any real $\lambda$ will be
considered). The white noise $\dot{w}(x)$ can be understood as the limit of a stationary Gaussian process $\xi(x)$
with independent increments on intervals of length $\Delta, \; \Delta \ll 1$ (cf. \cite{GMV11}):
%\end{document}
\begin{equation}
 \dot{w}(x) = \lim_{\Delta\to 0
}\frac{\xi_n}{\sqrt{\Delta}}, \quad x \in [n\Delta,(n+1)\Delta), 
\label{noise}
\end{equation}
and $\xi_n$ are independent identically distributed  ${\cal N} (0,1)$ random variables. An equivalent alternative approach to define the operator with a white noise
potential (which is used in the present paper) is based on the reduction of
(\ref{H}) to a system of It\^o stochastic differential equations of the first order 
for the phase $\theta$ and amplitude $r$ of $\psi$. In what follows we will denote 
$z = -\cot \theta$. Then the It\^o and Stratonovich representations of \rf{H} coincide
in coordinates $(z,r)$ and both give the same result (see \cite{PF92}).

After the standard substitution (in the Pr\"{u}fer form)
\[
u=\psi,~ v= \frac{\psi^\prime}{\omega},~~\omega=\sqrt{|\lambda|},
\]
equation  (\ref{H}) takes the form
\begin{equation}\label{sfo}
\begin{array}{l}
du = \omega v \,dx, \\
dv = \dst -\omega u\,dx + \frac{\sigma u \,dw}{\omega}.
\end{array}
\end{equation}
Usually the reduction of the Schr\"{o}dinger equation to a system of two
first order equations is accompanied by the phase-amplitude
formalism, i.e., by using the polar coordinates
\begin{equation}
\begin{array}{l}
u=\psi = r (x) \sin \theta (x), \\[2mm]
\dst v=\frac{\psi^\prime}{{\omega} } =  r (x) \cos \theta (x),
\end{array}
\label{pol}
\end{equation}
where $\theta \in [0,\pi)$.

The Lyapunov exponent $\gamma (\lambda ,\sigma )$ is defined by
\begin{equation}
\gamma (\lambda ,\sigma )=\lim_{t\rightarrow \infty }\frac{\ln
r(x)}{x}.%
\label{gamma}
\end{equation}%

This limit exists {\it P}-a.s. on the probability space $(\Omega, {\cal F}, {\cal P})$
defined by the white noise $\dot{w}(\cdot)$ in \rf{H} (see \cite{PF92}, \cite{CL90}).
We will discuss the limit behaviour of $%
\gamma (\lambda ,\sigma )$ in the different regimes when $\sigma \rightarrow
0$ or $\sigma \rightarrow \infty $ (small and large disorder) and $\sqrt{%
|\lambda| }=\omega \rightarrow 0$ or $\sqrt{|\lambda| }=\omega \rightarrow
\infty $ (long and short waves). We also introduce some combinations of $%
\lambda $ and $\sigma $ to describe the intermediate asymptotics.

Asymptotic behaviour of $\gamma$ for $\sigma \to 0$ and fixed $\lambda > 0$ was
established in papers \cite{APW86} and \cite{P86} where it was shown that 
\begin{equation}
\gamma \sim \frac{\sigma^2}{8\omega^2}, \quad \sigma \to 0.
\label{gam0}
\end{equation}
Paper \cite{PW88} contains the formula $\gamma (0,\sigma) \sim \sigma^{\frac{2}{3}}$ 
for fixed $\lambda = 0$ and $\sigma \to \infty$.
We show that in fact $\gamma(\omega,\sigma) = \omega f(\nu)$, where $\nu =
\frac{2\omega^3}{\sigma^2}$ and $f$ is a function of $\nu$. Thus, \rf{gam0} is
also valid for arbitrary $\sigma$ and $\omega$ such that $\nu \to \infty.$ 

In this paper we suggest a simple approach to obtaining the asymptotic expansion
for $\gamma$ as $\nu \to \infty$ as well as $\nu \to 0$. Therefore, different
combinations of large and small perturbations and frequencies are covered. We
show that $\gamma = C \sigma^{\frac{2}{3}}, \; \nu \to 0$. In particular, it provides
low frequency asymptotics of $\gamma$ when $\sigma$ is fixed and $\lambda \to
0$. The case of $\lambda < 0$ is also considered. The next section contains an exact
formula for $\gamma$ followed by a section with derivation of its asymptotics
expansion.
We will use equations (\ref{sfo}) to obtain the integral representations for $\gamma$.
\begin{theorem}
The following formulas for the Lyapunov exponent hold P-a.s.
\begin{enumerate}
 \item[1.]
If $\lambda>0$ then
\begin{equation}\label{gamma1}
 \gamma(\omega,\sigma) = \omega \nu^{-1}\int_{-\infty}^{\infty }\ipd(x)\,\frac{1-x^2}{(1+x^2)^2}\, dx,
\end{equation}
where
\begin{equation}
 \ipd(x) = C e^{\Phi(x)}\int_x^{\infty} e^{-\Phi(t)} \,dt, \quad \Phi(x) = \nu \left(x + \frac{x^3}{3}\right), \quad \nu = \frac{2\omega^3}{\sigma^2},
\label{p}
\end{equation}
with a normalizing constant $C$ given by
\begin{equation}
 C^{-1} = \sqrt{\frac{2\pi}{\nu}}\int_0^\infty \frac{e^{-2\Phi(x) }}{\sqrt{x}}\,dx.
\label{Cp}
\end{equation}
\item[2.]
If $\lambda<0$ then
\begin{equation}
 \gamma(\omega,\nu) = \omega \nu^{-1}\int_{-\infty}^{\infty }\ipd(x)\,\frac{1-x^2}{(1+x^2)^2}\, dx - 2\omega \int_{-\infty}^{\infty }\ipd(x)\,\frac{x}{1+x^2}\, dx.
\label{gamma-}
\end{equation}
where
\begin{equation}
 \ipd(x) = C e^{\Psi(x)}\int_x^{\infty} e^{-\Psi(t)} \,dt, \quad \Psi(x) = \nu \left(\frac{x^3}{3}-x\right), \quad \nu = \frac{2\omega^3}{\sigma^2},
\label{p-}
\end{equation}
and
\begin{equation}
 C^{-1}  = \sqrt{\frac{2\pi}{\nu}}\int_0^\infty \frac{e^{-2 \Psi(x)}}{\sqrt{x}}\,dx.
\label{c-}
\end{equation}
\end{enumerate}
\end{theorem}
\begin{proof}
Let $\lambda>0.$ We rewrite (\ref{sfo}) in polar coordinates \rf{pol} but use $z=-\cot\theta$ instead of $\theta$. The resulting equations will be much simpler in the coordinates $(z,r)$ than in $(r,\theta)$. One only needs to keep in mind that the process  $z(x)=-\cot\theta(x)$ has boundary conditions at infinity: if $z(x)\to\infty$ when $x$ approaches from the left to some finite value $x_0$, then $z(x_0+0)=-\infty$. This jump corresponds to the transition of the angle variable $\theta=\theta(x)$ from $\pi$ to zero.
We will use the It\^o formula in order to write equations for $(z,r)$. The It\^o formula says that for any smooth function $F=F(u,v)$,
\begin{equation}
 dF(u(x),v(x)) = \frac{\pa F}{\pa u} du + \frac{\pa F}{\pa v} dv
 + \frac{1}{2}\frac{\pa^2 F}{\pa v^2} dv^2,
\end{equation}
where $du,~dv$ are given by (\ref{sfo}) and $dw^2=dx$. There is only one quadratic in $du,~dv$ term
on the right, since equation \rf{sfo} for $du$ does not contain $dw$. If we put $du,~dv$ from (\ref{sfo}) into the It\^o formula we arrive at
\begin{align}
dF(u,v) &=  \frac{\pa F}{\pa u} \omega vdx + \frac{\pa F}{\pa v} \left( \frac{\sigma u \,dw}{\omega} - \omega u\,dx \right)
 + \frac{1}{2}\frac{\pa^2 F}{\pa v^2} \frac{\sigma^2 u^2}{\omega^2}dx \nonu \\[2mm]
&= \left(\frac{\pa F}{\pa u}\omega  v-\frac{\pa F}{\pa v} \omega u + \frac{1}{2}\frac{\pa^2 F}{\pa v^2} \frac{\sigma^2 u^2}{\omega^2}  \right)dx + \frac{\sigma}{\omega} \frac{\pa F}{\pa v} udw .
\label{genf}
\end{align}

Since $\dst z=-\cot \theta=-\frac{v}{u}$ (see (\ref{pol})), we apply (\ref{genf}) to function $\dst F(u,v)=z=-\frac{v}{u}$. For this function $F$ we have 
\begin{align}
\frac{\pa F}{\pa u} = \frac{v}{u^2},\quad
 \frac{\pa F}{\pa v}= -\frac{1}{u}, \quad
\frac{\pa^2 F}{\pa v^2} = 0,
\end{align}
and therefore,
\begin{equation}\label{eqz}
dz= \omega( 1+z^2)\,dx-\frac{\sigma}{\omega}\,dw. 
\end{equation}

We will need function $\ln r=\frac{1}{2}\ln (u^2+v^2)$, not $r$. Thus from (\ref{genf}) with $F(u,v)=\frac{1}{2}\ln (u^2+v^2)$ we obtain
\begin{equation}\label{eqln}
 d\ln r=\frac{\sigma^2 }{2\omega^2}\frac{1-z^2}{(1+z^2)^2}\,dx-\frac{\sigma}{\omega}\frac{z}{1+z^2}\,dw.
\end{equation}

The right-hand side of the latter equation does not depend on $r$. Thus the solution of the original Schr\"{o}dinger  equation (\ref{H}) is reduced to a solution of the stochastic equation (\ref{eqz}) followed by a direct integration of (\ref{eqln}).

Consider the generator of the diffusion process $z(x)$:
\[
\mathcal L =\mathcal L_z =\frac{\sigma^2 }{2\omega^2}\frac{\pa^2 }{\pa z^2}+\omega( 1+z^2)\frac{\pa}{\pa z}.
\]
The probability density $p(x,z_1,z_2)$ (with respect to $z_2$ with initial point $z_1$) of the diffusion process $z(x)$ satisfies the equation
\begin{equation}
 \frac{\partial p}{\partial x}= \L_{z_1} p = \L^*_{z_2} p\quad {\text {with}} \quad \left. p \right|_{x=0} = \delta(z_1 - z_2),
\end{equation}
and the limiting probability density $\ipd$ is given by
\begin{equation}
 \L^\ast \,\ipd = 0
\end{equation}
or
\begin{equation}
 \frac{\sigma^2}{2\omega^3}\,\frac{d^2 \ipd}{d z^2} - \frac{d}{d z}\left((1+z^2)\,\ipd \right)=0.
\label{dep}
\end{equation}
The solution of the last equation has the form (\ref{p}) where the normalizing constant is defined by the condition
\begin{equation}
\int_{-\infty}^\infty
\ipd(z)dz=1,
\label{norm_p}
\end{equation}
that is,
\[
C^{-1} = \int_{-\infty}^\infty  e^{\Phi(z)} \,dz\int_z^{\infty} e^{-\Phi(\tau)}\,d\tau .
\]
We make the substitution $\tau=s+z$ here and evaluate the integral in $z$. This leads to (\ref{Cp}). After the limiting probability density $\ipd(z)$ is found, formula \rf{gamma1} is an immediate consequence of (\ref{gamma}) and (\ref{eqln}). One needs only to note that the boundedness of the integrand in the second term in (\ref{eqln}) implies that $P$-a.s.
\[
\frac{1}{x}\int_0^x\frac{z}{1+z^2}dw \to 0 \quad {\text{as}} ~~ x\to\infty. 
\]

If $\lambda \leq 0$ then we denote $\omega^2 = -\lambda$ in \rf{H}, and the system \rf{sfo} becomes
\begin{equation}\label{sfo1}
\begin{array}{l}
du = \omega v \,dx, \\
dv = \dst \omega u\,dx + \frac{\sigma u \,dw}{\omega}.
\end{array}
\end{equation}
The stochastic differential equation for $z(x)=-\cot \theta$ analogous to \rf{eqz} has the form 
\begin{equation}\label{eqz1}
dz= \omega(z^2-1)\,dx-\frac{\sigma}{\omega}\,dw,
\end{equation}
The stochastic equation for $\ln r$ similar to \rf{eqln} is
\begin{equation}\label{eqln1}
 d\ln r=\left[\frac{\sigma^2 }{2\omega^2}\frac{1-z^2}{(1+z^2)^2} -\frac{2\omega z}{1+z^2} \right] dx
-\frac{\sigma}{\omega}\frac{z}{1+z^2}\,dw.
\end{equation}
Determining the limiting probability density $\ipd$ as in \rf{dep}
we obtain that it satisfies the equation
\begin{equation}
 \frac{\sigma^2}{2\omega^3}\,\frac{d^2 \ipd}{d z^2} - \frac{d}{d z}\left((z^2-1)\,\ipd \right)=0.
\label{dep1}
\end{equation}
The solution of this equation is given by \rf{p-} with a normalizing constant
determined from \rf{c-}. Therefore, the expression of the Lyapunov exponent
is given by \rf{gamma-}.
\end{proof}

\section{Asymptotics of the Lyapunov exponent}
\setcounter{equation}{0}

The exact expression of the Lyapunov exponent $\gamma$ is given by \rf{gamma1} if $\lambda \geq 0$ and by \rf{gamma-} if $\lambda \leq 0$, and can be calculated numerically. Although equation \rf{H} depends on two parameters $\lambda$ and $\sigma$, formula \rf{gamma1} implies that $\gamma$ depends on the product of $\omega=\sqrt{|\lambda|} $ and a function of $\nu = \frac{2\omega^3}{\sigma^2}$. The next theorem provides the asymptotic behaviour of the Lyapunov exponent for $\nu \to 0$ and $\nu \to \infty$.

\begin{theorem}
Lyapunov exponent of the Schr\"{o}dinger operator \rf{H} has the following asymptotic behaviour:
 \begin{itemize}
\item[(a)] If $\lambda \gtrless 0$ and $\nu \to 0$ (long wave asymptotics) then
\begin{equation}
 \gamma = c\,\sigma^{\frac{2}{3}}\left(1 + O(\nu^{\frac{2}{3}})\right),
 \label{gamma0}
\end{equation}
where exact value of $c$ is given by \rf{c}.
\item[(b)] If $\lambda >0$ and $\nu  \to \infty$ (short wave asymptotics) then
\begin{equation}
\gamma =\frac{\sigma^2}{8\omega^2 }\left( 1-\frac{15}{16}\,\nu ^{-2}+O(\nu ^{-4})\right).
\label{gamma_h}
\end{equation}
\item[(c)] If $\lambda <0$ and $\nu  \to \infty$ then
\begin{equation}
\gamma = \omega(1+O(\nu^{-1})).
\label{gamma_neg}
\end{equation}
\end{itemize}
\end{theorem}

\begin{proof}
\begin{itemize}
\item[(a)] 
Let us find the asymptotic expansion of $C^{-1}$ in \rf{Cp} using the substitution $t=2\nu x^3/3$. If $\lambda > 0$ then 
\begin{align}
  C^{-1} &= \sqrt{\frac{2\pi}{\nu}}\int_0^\infty \frac{e^{-2\nu \left(x+\frac{x^3}{3} \right)}}{\sqrt{x}}\,dx=\frac{1}{3}\sqrt{\frac{2\pi}{\nu}}
\left( \frac{3}{2\nu}\right)^{\frac{1}{6}} \int_0^\infty t^{-\frac{5}{6}} e^{-t-(12\nu^2 t)^{\frac{1}{3}}}\,dt  \nonu \\[2mm]
& = \sum_{n=0}^\infty (-1)^n a_n\, \nu^{\frac{2(n-1)}{3}}, \quad
a_0 = \pi^\oh \,2^{\frac{1}{3}}\, 3^{-\frac{5}{6}}\, \Gamma \left(\frac{1}{6} \right), \quad \nu \to 0,
\label{a}
\end{align}
where $a_n >0$ and $\Gamma(x)$ is the Gamma function. Hence, for the constant $C$ we obtain
\begin{equation}
C = a_0\, \nu^{\frac{2}{3}} \left(1 + O(\nu^{\frac{2}{3}})\right), \quad \nu \to 0.
\label{C1}
\end{equation}

To find asymptotics of $\gamma$ we use two substitutions
$x=\nu^{-\frac{1}{3}}s$ and $t=\nu^{-\frac{1}{3}}y$ in \rf{gamma1} 
and \rf{p}, respectively. Then
\begin{align}
 \gamma &= \omega \nu^{-1}C\int_{-\infty}^{\infty}
 e^{\nu^{\frac{2}{3}}s+\frac{s^3}{3}} \frac{\nu^\frac{2}{3}
 -s^2}{(\nu^\frac{2}{3}+s^2)^2}\,ds \int_s^\infty
 e^{-\nu^{\frac{2}{3}}y-\frac{y^3}{3}} dy \nonu \\[2mm]
 & = \omega \nu^{-1} C\int_{-\infty}^{\infty}
 \frac{\nu^\frac{2}{3}-s^2}{(\nu^\frac{2}{3}+s^2)^2}\, Q(s,\nu)\,ds, 
 \label{gamma11}
\end{align}
where 
\begin{align}
 Q(s,\nu) = \int_s^\infty
 e^{-\nu^{\frac{2}{3}}(y-s)-\frac{1}{3}(y^3-s^3)}\, dy 
 = \int_0^\infty e^{-\nu^{\frac{2}{3}}t-\frac{1}{3}[(t+s)^3-s^3]}\, dt.
 \label{Q1}
\end{align}
Denote by $L$ a contour in the complex $s$-plane which consists of
the real axis with the interval $(-1,1)$ replaced by the semicircle $|s|=1$, $\pi \geq \arg s \geq 0$. Then for $\nu < 1$
\begin{equation}
 \gamma = \omega \nu^{-1} C\left\{\int_{L}
 \frac{\nu^\frac{2}{3}-s^2}{(\nu^\frac{2}{3}+s^2)^2}\, Q(s,\nu)\,ds
 + 2\pi i\res_{s = i\nu^\frac{1}{3}} \frac{\nu^\frac{2}{3}-s^2}{(\nu^\frac{2}{3}+s^2)^2}\, Q(s,\nu)  \right\}.
 \label{cont}
\end{equation}
One can show that $Q$ and each of its derivatives with respect to $\mu=\nu^\frac{2}{3}$ are uniformly bounded in $\nu \leq 1$, $s \in \R$. Thus, the first term in \rf{cont} is infinitely smooth function of $\mu$, $\mu < 1$, and
\begin{align}
 &\dst \int_{L}\frac{\nu^\frac{2}{3}-s^2}{(\nu^\frac{2}{3}+s^2)^2}\, Q(s,\nu)\,ds 
 = -\int_{L}\frac{Q(s,0)}{s^2}\,ds + O(\nu^\frac{2}{3}), \quad \nu \to 0.
 \label{main}
\end{align}
The main term in \rf{main} can also be written in the form
\begin{align}
 -\int_{L}\frac{Q(s,0)}{s^2}\,ds
 & = \dst \pv_{-\infty}^\infty \frac{Q(0,0) - Q(s,0)}{s^2}\,ds 
 + \pi i\, Q^\prime_s (0,0) \nonu \\[2mm]
 & = \dst \pv_{-\infty}^\infty \frac{Q(0,0) - Q(s,0)}{s^2}\,ds 
 - \pi i,
\end{align}
where $\int\hspace{-8pt}\sst{-}$ stands for the Cauchy principle value. The second term in \rf{cont} is analytic in $\nu^{\frac{1}{3}}$ and hence
\begin{align}
  &2\pi i\res_{s = i\nu^\frac{1}{3}} \frac{\nu^\frac{2}{3}-s^2}{(\nu^\frac{2}{3}+s^2)^2}\, Q(s,\nu)
  =-\pi i Q^\prime_s (i\nu^{\frac{1}{3}},\nu) \nonu \\[2mm]
  &= -\pi i \left[a + b\nu^{\frac{1}{3}} +
  O(\nu^\frac{2}{3}) \right],
  \label{}
\end{align}
where $\dst a = -\int_0^\infty t^2 e^{-\frac{1}{3}t^3}\, dt = -1$ and
$\dst b = i\int_0^\infty t(t^3-2) e^{-\frac{1}{3}t^3}\, dt =0$.
Thus, we obtain two formulas for computation of $\gamma$: through the principal value and as a contour integral
\begin{align}
 \gamma &= \dst \omega \nu^{-1} C \left\{
 \pv_{-\infty}^\infty \frac{Q(0,0)-Q(s,0)}{s^2}\,ds + O(\nu^\frac{2}{3})\right\} 
 \nonu \\[2mm]
 &=\dst \omega \nu^{-1}C \left\{-\int_{L}\frac{Q(s,0)}{s^2}\,ds + \pi i + O(\nu^\frac{2}{3})\right\} \nonu \\[2mm]
 &=c\, \sigma ^{\frac{2}{3}}\left(1 + O(\nu^{\frac{2}{3}})\right), \quad \nu \to 0,
 \label{gamma_a}
\end{align}
where 
\begin{equation}
c = \dst \frac{3^{\frac{5}{6}}}{\pi^{\frac{1}{2}}\,2^{\frac{2}{3}}\,\Gamma \left( \frac{1}{6}\right)}\,\pv_{-\infty}^\infty \frac{Q(0,0)-Q(s,0)}{s^2}\,ds. 
\label{c}
\end{equation}

If $\lambda < 0$ then formula \rf{c-} for $C^{-1}$ differs from 
\rf{Cp} only by a sign in the exponent. 
Therefore expansion \rf{a} with $(-1)^n a_n$ replaced by $a_n$ holds:
\begin{align}
 C^{-1} &= \dst \frac{1}{3}\sqrt{\frac{2\pi}{\nu}}\left(\frac{3}{2\nu} \right)^{\frac{1}{6}}\int_0^\infty t^{-\frac{5}{6}} e^{-t+(12\nu^2 t)^{\frac{1}{3}}}\,dt %\nonu \\[2mm]
 = \dst \sum_{n=0}^\infty a_n \,\nu^{\frac{2(n-1)}{3}}, \quad \nu \to 0.
\end{align}
Thus,  asymptotics \rf{C1} is valid in this case.

The Lyapunov exponent is determined from \rf{p-} and \rf{gamma-}:
\begin{align}
 \gamma &= \dst \frac{\sigma ^{2}}{2\omega ^{2}}\,C\int_{-\infty}^{\infty }\frac{1-x^2}{(1+x^2)^2}\,e^{\nu(\frac{x^3}{3}-x)}\, dx\int_x^{\infty} e^{-\nu(\frac{t^3}{3}-t)}\, dt \nonu \\[2mm]
& \dst -2\omega \,C\int_{-\infty}^{\infty }\frac{x e^{\nu(\frac{x^3}{3}-x)}}{1+x^2}\, dx\int_x^{\infty} e^{-\nu(\frac{t^3}{3}-t)}\, dt  \nonu\\[2mm]
&\dst =\omega \nu^{-1} C\int_{-\infty}^{\infty }\frac{\nu^\frac{2}{3}-x^2}{(\nu^{\frac{2}{3}}+x^2)^2}\, R(x,\nu) \, dx 
 -2\omega \nu^{-\frac{1}{3}}C \int_{-\infty}^{\infty }\frac{xR(x,\nu)\,dx}{\nu^\frac{2}{3}+x^2},
\label{gamma31}
\end{align}
where 
\begin{align}
 R(x,\nu) = \int_x^{\infty} e^{\nu^\frac{2}{3}(t-x)-\frac{1}{3}(t^3 -x^3)}\, dt = \int_0^{\infty} e^{\nu^\frac{2}{3}y-\frac{1}{3}((y+x)^3 -x^3)}\, dy.
\end{align}
Denote the terms on the right hand side of \rf{gamma31} by $I_1$ and $I_2$, respectively.
Then $I_1$ is analogous to that in \rf{gamma11} and gives the main contribution to $\gamma$. To determine asymptotics of $I_1$, we use the same contour $L$ in the complex $s$-plane as in \rf{cont}:
\begin{align}
\fl I_1 & = \omega \nu^{-1}C\left\{\int_{L} \frac{\nu^\frac{2}{3}-s^2}{(\nu^{\frac{2}{3}}+s^2)^2}\, R(s,\nu) \, ds + 2\pi i\res_{s = i\nu^\frac{1}{3}} \frac{\nu^\frac{2}{3}-s^2}{(\nu^\frac{2}{3}+s^2)^2}\, R(s,\nu) 
  \right\} \nonu \\[2mm]
\fl & = \omega \nu^{-1}C\left\{-\int_{L} \frac{R(s,0)}{s^2}\, ds - \pi i  R^\prime_x (i\nu^{\frac{1}{3}},\nu) + O(\nu^\frac{2}{3}) \right\} \nonu \\[2mm]
\fl & = \omega \nu^{-1}C\left\{-\int_{L} \frac{Q(s,0)}{s^2}\, ds + \pi i + O(\nu^\frac{2}{3}) \right\}  \nonu \\[2mm] 
\fl & = \omega \nu^{-1}C\left\{\pv_{-\infty}^{\infty} \frac{Q(0,0)-Q(s,0)\,ds}{s^2} + O(\nu^\frac{2}{3})\right\} 
=c\, \sigma ^{\frac{2}{3}}\left(1 + O(\nu^{\frac{2}{3}})\right), \quad \nu \to 0,
\end{align}
where $c$ is defined by \rf{c}. The second term $I_2$ in \rf{gamma31} has smaller order
in $\nu$:
\begin{align}
\fl I_2 & = -2\omega \nu^{-\frac{1}{3}}C \left\{\int_{L}\frac{sR(s,\nu)\,ds}{\nu^\frac{2}{3}+s^2}  + 2\pi i\res_{s = i\nu^\frac{1}{3}} \frac{\nu^\frac{2}{3}-s^2}{(\nu^\frac{2}{3}+s^2)^2}\, R(s,\nu) 
  \right\} \nonu \\[2mm]
\fl & = -2\omega \nu^{-\frac{1}{3}}C\left\{\int_{L} \frac{R(s,0)}{s}\, ds + \pi i  R (i\nu^{\frac{1}{3}},\nu) + O(\nu^\frac{2}{3}) \right\} \nonu \\[2mm]
\fl & = -2\omega \nu^{-\frac{1}{3}}C\Biggl\{\pv_{-\infty}^\infty \frac{R(s,0)}{s}\, ds - \pi i\res_{s = i\nu^\frac{1}{3}} \frac{R(s,0)}{s} 
+ \pi i  R (i\nu^{\frac{1}{3}},\nu)+ O(\nu^\frac{2}{3}) \Biggr\}  \nonu \\[2mm] 
\fl  &= -2\omega \nu^{-\frac{1}{3}}C\left\{\pv_{-\infty}^{\infty} \frac{Q(s,0)}{s}\,ds 
+ O(\nu^\frac{2}{3})\right\}=\sigma^\frac{2}{3} O(\nu^\frac{2}{3}), \quad \nu \to 0. 
\end{align}
Thus, asymptotics \rf{gamma_a} is valid for $\lambda < 0$ as well.

\item[(b)] 
Using the substitution $t=2\nu x$ we compute asymptotics of $C^{-1}$ in \rf{Cp}
\begin{align}
 C^{-1} &= \frac{\sqrt{\pi}}{\nu}\int_0^\infty t^{-\oh} e^{-t-\frac{t^3}{12\nu^2}}\,dt 
= \sum_{n=0}^\infty b_n \nu^{-2n-1}, \\[2mm]
 b_n &= \frac{(-1)^{n}}{n!} \sqrt{\pi}\, \frac{\Gamma
\left( 3n+\frac{1}{2}\right) }{12^{n}}, \quad \nu \to \infty.
\end{align}
Therefore, for asymptotics of $C$ we have
\begin{equation}
C=\frac{\nu}{\pi }\left( 1+\frac{5}{32}\nu ^{-2}+O\left( \nu ^{-4}\right)
\right) ,\quad \nu \to \infty.  
\label{c_asy_inf}
\end{equation}
Next, we find asymptotics of \rf{p} using repeated integration by parts:
\begin{align}
 &\int_x^{\infty }e^{-\nu \left( t+ \frac{t^3}{3}\right)}dt
 =\int_x^{\infty }\frac{d e^{-\nu \left(t+ \frac{t^3}{3}\right)}}{-\nu (1+t^{2})} \nonu \\[2mm]
 &=\frac{e^{-\nu \left( x+ \frac{x^3}{3}\right)}}{\nu (1+x^{2})}\Biggl(
1-\frac{2x}{\nu (1+x^{2})^{2}} +\frac{2(5x^{2}-1)}{\nu ^{2}(1+x^{2})^{4}}%
+O(\nu ^{-3})\Biggr) ,\quad \nu \to \infty.
\label{int_asy_inf}
\end{align}
This asymptotics is uniform in $x$. Hence, from \rf{gamma1}, \rf{p}, \rf{c_asy_inf}, and \rf{int_asy_inf} we obtain
\begin{align}
\fl \gamma &=\frac{\sigma^2}{2\nu \omega^2} \,C\int_{-\infty }^{\infty }\frac{1-x^{2}}{(1+x^{2})^{3}}%
\left( 1-\frac{2x}{\nu (1+x^{2})^{2}}+\frac{2(5x^{2}-1)}{\nu
^{2}(1+x^{2})^{4}} +O(\nu ^{-3}) \right) dx \quad \nonu \\[2mm]
\fl &=\frac{\sigma^2}{8\omega^2 }\left( 1-\frac{15}{16}\,\nu ^{-2}+O(\nu ^{-4})\right), \quad \nu \to \infty.
\end{align}

\item[(c)] 
The exponent in the integrand \rf{c-} has a maximum at $x=1$. Therefore, Laplace's method \cite{F} immediately gives 
\begin{align}
 C^{-1} = \frac{\pi}{\nu}\,e^{\frac{4\nu}{3}}(1+O(\nu^{-1})),\quad \nu \to \infty.
\label{C-}
\end{align}
We rewrite the Lyapunov exponent  as follows:
\begin{align}
 \gamma &=\omega \nu^{-1}C\int_{-\infty}^{\infty }\frac{1-x^2}{(1+x^2)^2}\,e^{\nu(\frac{x^3}{3}-x)}\, dx\int_x^{\infty} e^{-\nu(\frac{t^3}{3}-t)}\, dt \nonu \\[2mm]
& -2\omega \,C\int_{-\infty}^{\infty }\frac{x e^{\nu(\frac{x^3}{3}-x)}}{1+x^2}\, dx\int_x^{\infty} e^{-\nu(\frac{t^3}{3}-t)}\, dt.
\label{gamma41}
\end{align}

Observe that the maximum of the exponent in the second integral is attained
at $t=1,x=-1$. Therefore, Laplace's method applied to the double integral \rf{gamma41} (see \cite{F} ) together with \rf{C-} gives
\begin{align}
& -2\omega \,C\int_{-\infty}^{\infty }\frac{x e^{\nu(\frac{x^3}{3}-x)}}{1+x^2}\, dx\int_x^{\infty} e^{-\nu(\frac{t^3}{3}-t)}\, dt \nonu \\[2mm]
&= 2\omega \,C \frac{\pi}{2\nu}\,e^{\frac{4\nu}{3}}\, (1+O(\nu^{-1})) = \omega (1+O(\nu^{-1})),\quad \nu \to \infty.
\end{align}
Since the integrand of the first term in \rf{gamma41} vanishes at $x=-1$, the same approach
leads to the asymptotics
\begin{align}
\fl \omega \nu^{-1}C\int_{-\infty}^{\infty }\frac{1-x^2}{(1+x^2)^2}\,e^{\nu(\frac{x^3}{3}-x)}\, dx\int_x^{\infty} e^{-\nu(\frac{t^3}{3}-t)}\, dt 
 = O(\nu^{-2}),\quad \nu \to \infty.%\nonu \\[2mm]
\end{align}
Hence, in this case we have
\begin{align}
\gamma = \omega(1+O(\nu^{-1})).
\end{align}

\end{itemize}
\end{proof}

\section{Lyapunov exponent of short-range correlated noise}
\setcounter{equation}{0}

By short-range correlated (or real) noise we understand a piecewise-constant process 
\begin{equation}
 \xi_{\Delta}(x) = \frac{\xi_n}{\sqrt{\Delta}}, \quad 
 x \in [n\Delta, (n+1)\Delta], \quad n=0,1,2, \ldots,
\end{equation}
where $\xi_n$ are i.i.d. random variables \rf{noise} and $\Delta$ is a small
finite number (correlation length).
Introducing the polar coordinates
\begin{equation}
\begin{array}{l}
\psi (n\Delta) = r (n\Delta) \sin \theta (n\Delta), \\[2mm]
\psi^\prime (n\Delta) = \omega r (n\Delta) \cos \theta (n\Delta),
\end{array}
\label{pol1}
\end{equation}
and relating the values of $\psi \left((n+1)\Delta\right)$ and $\omega^{-1} \psi^\prime ((n+1)\Delta)$ through the transfer matrix (see \cite{GMV11}) one can find
\begin{align}
\label{ph1}
 z((n+1)\Delta) &=z(n\Delta)- \frac{\sigma \xi_n\sqrt{\Delta}}{\omega}+\omega(1+z^2(n\Delta))\Delta-\sigma \xi_nz(n\Delta)\Delta^{3/2} \nonu \\[2mm] 
 &+\left(\frac{\sigma^2 \xi^2_n}{3\omega}+\left(1+z^2(n\Delta)\right)\omega^2 z(n\Delta)\right)\Delta^2
 +O(\Delta^{5/2}), \\[2mm]
 r^2((n+1)\Delta) &=r^2(n\Delta) \biggl(1+\frac{\sigma \xi_n \sqrt{\Delta}\sin 2\theta(n\Delta)}{\omega}+\frac{\sigma^2 \xi_n^2 \Delta}{\omega^2}\sin^2\theta(n\Delta) \nonu \\[2mm]
 &+\sigma\xi_n \Delta^{3/2}\cos 2\theta(n\Delta) +\frac{2\sigma^2\xi_n^2\Delta^2\sin(2\theta(n\Delta))}{3\omega} \biggr)+O(\Delta^{5/2}),
 \label{amp1}
\end{align}
where $z(n\Delta) = \cot \theta (n\Delta)$. For small $\Delta$ equations \rf{ph1}-\rf{amp1} can be approximated by stochastic differential equations
\begin{align}
\label{ph2}
 dz &=- \frac{\sigma}{\omega}\,dw + \omega(1+z^2)dx 
 +\Delta\left(\frac{\sigma^2}{3\omega}+\left(1+z^2\right)\omega^2 z\right)dx, \\[2mm]
 d\ln r &= -\frac{\sigma}{\omega}\frac{z}{1+z^2}\,dw+\frac{\sigma^2}{2\omega^2} \frac{1-z^2}{(1+z^2)^2}\,dx \nonu \\[2mm]
 &+\Delta\left[\frac{\sigma^2}{3\omega}\frac{z(z^2 -5)}{(1+z^2)^2} 
 + \frac{\sigma^2}{2\omega^4}\frac{z^2(z^2 -3)}{(1+z^2)^4}
 + \frac{3\sigma^4}{4\omega^4}\frac{2z^2 -1}{(1+z^2)^2} \right]dx
 \label{amp2}
\end{align}
We want to find asymptotics of the solutions of \rf{ph2}-\rf{amp2} under assumptions that $\sigma \ll 1$, 
$\Delta \ll 1$, and $\Delta \omega \ll 1$. Then we can drop the term $\sigma^2/(3\omega)$ in \rf{ph2} and using previously described procedure determine the limiting probability density $p(z)$ corresponding to the equation
\begin{equation}
 dz =- \frac{\sigma}{\omega}\,dw + \omega(1+z^2)(1+\Delta \omega z)dx.
 \label{ph3}
\end{equation}
Derivation of \rf{ph2}-\rf{ph3} is based on the Taylor expansion of \rf{ph1}-\rf{amp1} and has physical level of rigor. Subsequent solution of these equations, however, is rigorous mathematically.

As in \rf{dep}, $p(z)$ satisfies the equation
\begin{equation}
 \frac{d^2 p}{dz^2} - \nu \frac{d}{dz}\left( (1+z^2)(1+\Delta \omega z) p\right) = 0,
 \quad \nu = \frac{2\omega^3}{\sigma^2},
 \label{p_noise}
\end{equation}
subject to normalization condition \rf{norm_p}.
Solution of \rf{p_noise} can be represented as a power series in $\Delta$
\begin{equation}
 p(z) = p_0(z) + p_1(z) \Delta + p_2(z) \Delta^2 + \ldots,
\end{equation}
where $p_0(z)$ is given by \rf{p} and represents the limiting probability density
corresponded to the white noise. Function $p_1(z)$ satisfies the following equation
\begin{align}
 \frac{d^2 p_1}{dz^2} - \nu \frac{d}{dz}\left( (1+z^2) p_1 \right) = 
 \nu \omega \frac{d}{dz}\left( (1+z^2) z p_0\right)
 \label{p1_eq}
\end{align}
and the conditions
\begin{align}
 \int_{-\infty}^\infty p_1(z)\,dz=0, \quad p_1(z) \to 0 ~~ \text{as} ~~ z \to \pm \infty.
 \label{p1_con}
\end{align}
Solving \rf{p1_eq} with \rf{p1_con} we obtain
\begin{align}
 p_1(z) &= \nu \omega C e^{\Phi(z)} \int_z^\infty x(1+x^2)\,dx \int_x^\infty e^{-\Phi(t)}dt 
 + \nu \omega C^2 e^{\Phi(z)} \int_z^\infty e^{-\Phi(x)}dx \nonu \\[2mm]
 & \times \pv_{-\infty}^\infty e^{\Phi(z)}dz  \int_z^\infty x(1+x^2)\,dx \int_x^\infty e^{-\Phi(t)}dt. 
\end{align}
Integration by parts with \rf{c_asy_inf} gives asymptotic behaviour of $p_1(z)$
\begin{equation}
 p_1(z) = \frac{\omega}{\pi (1+z^2)} \left( z + \frac{1-3z^2}{\nu (1+z^2)^2} + O(\nu^{-2})\right), \quad \nu \to \infty.
\end{equation}
Hence, for a fixed frequency $\omega$ and a small correlation length $\Delta$ 
we obtain 
the following dependence of the Lyapunov exponent $\gamma_{\text r}$ on the noise intensity $\sigma$
\begin{align}
\gamma_{\text r} &= \int_{-\infty}^\infty \left[\frac{1}{\pi (1+z^2)}\left(1- \frac{2z}{\nu (1+z^2)^2}\right) +  \frac{\omega \Delta}{\pi (1+z^2)} \left( z + \frac{1-3z^2}{\nu (1+z^2)^2} \right) \right] \nonu \\[2mm]
&\times \left[\frac{\sigma^2}{2\omega^2} \frac{1-z^2}{(1+z^2)^2} + \Delta\left(\frac{\sigma^2}{3\omega}\frac{z(z^2 -5)}{(1+z^2)^2} 
 + \frac{\sigma^2}{2\omega^4}\frac{z^2(z^2 -3)}{(1+z^2)^4}\right) \right] dz \nonu \\[2mm]
 &= \frac{\sigma^2}{8\omega^2}\left( 1 - \frac{3\Delta}{8\omega^2} + O(\sigma^2) \right),
 \quad \sigma \to 0.
 \label{gamma_r}
\end{align}
Comparison of the Lyapunov exponents $\gamma$ and $\gamma_{\text r}$ for the 
stochastic oscillator driven by white noise process or a mean-zero function of an ergodic 
finite-state reversible Markov process, respectively,  has been done in \cite{PW92}-\cite{PW91} where it was qualitatively shown that $\gamma > \gamma_{\text r}$. Our result
gives also a quantitative estimate under the assumption that the short-range correlated noise is understood as a stationary Gaussian process with independent increments on intervals of length $\Delta \ll 1$.

\section{Numerical results}
\setcounter{equation}{0}

Since we do not estimate the remainders in asymptotic formulas \rf{gamma0}-\rf{gamma_neg},
we compare the formulas numerically with the exact ones given by \rf{gamma1}-\rf{gamma-} in order to determine the range of their validity.
Figures \ref{pdf1} and \ref{pdf2} show the graphs of the normalized Lyapunov exponent $\gamma/\omega$ \rf{gamma1} as a function of $\dst \nu = {2\omega^3}/{\sigma^2}, \;\; \omega = \sqrt{|\lambda|},$ for $\lambda > 0$. As $\nu \to 0$ in figure \ref{pdf1} the relative error of the asymptotic approximation of $\gamma$ becomes less than $0.7\%$ when $\nu < 10^{-3}$. The relative error of asymptotic approximation \rf{gamma_h} for large values of $\nu$ is less than $0.4\%$ if $\nu > 6$ in figure \ref{pdf2}.  
\begin{figure}[H]
\begin{center}
\includegraphics[width=0.6\textwidth, angle=0]{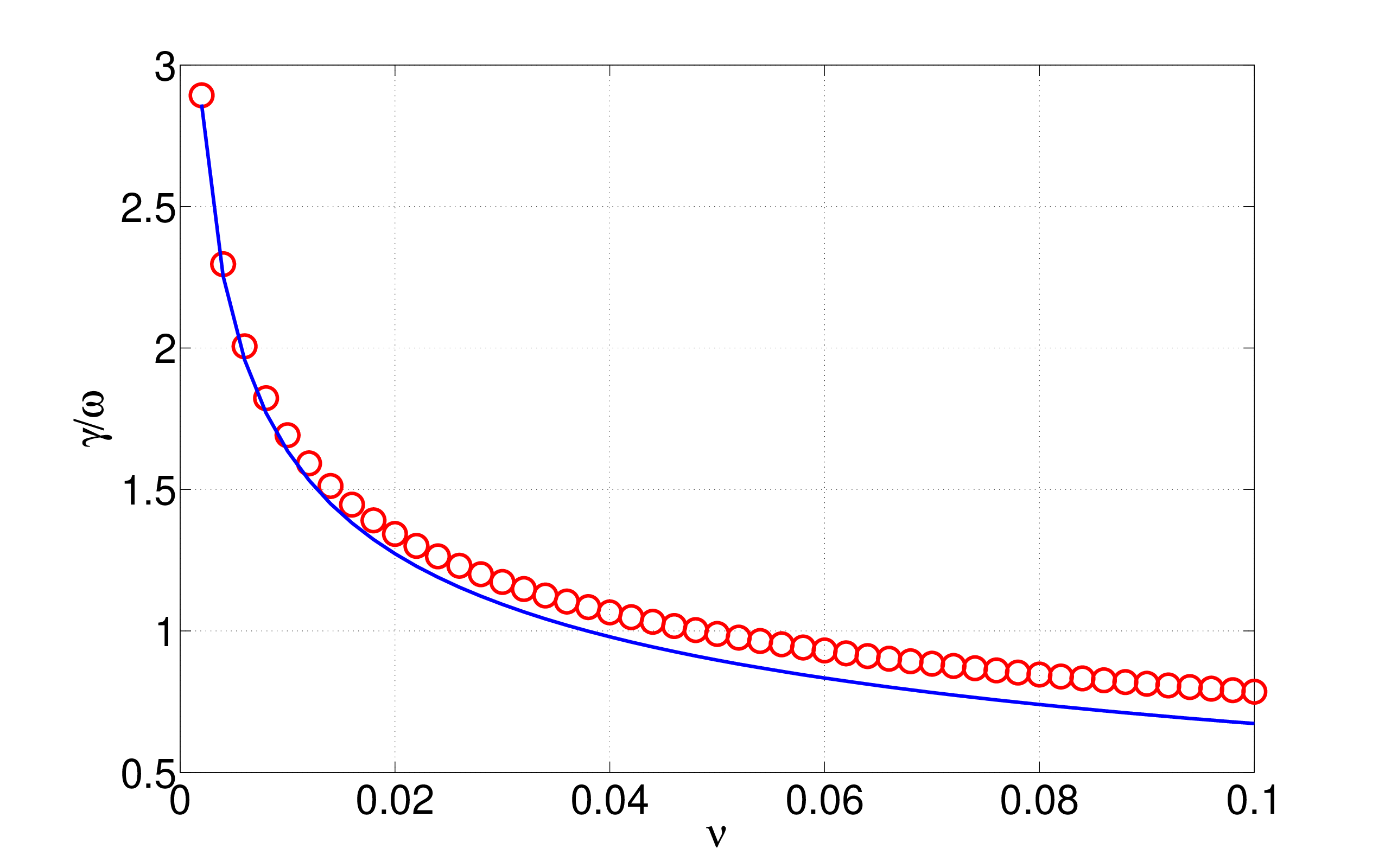}
\end{center}
\caption{Graph of the normalized Lyapunov exponent $\gamma/\omega$ of the Schr\"{o}dinger equation  with white noise potential (solid line) and its asymptotic approximation $\gamma/\omega = 0.3645 \nu^{-1/3}$ (circles) for small $\dst \nu = {2\omega^3}/{\sigma^2}, \; \omega = \sqrt{|\lambda|},$  and $\lambda > 0$.} 
\label{pdf1}
\end{figure}
\begin{figure}[H]
\begin{center}
\includegraphics[width=0.6\textwidth, angle=0]{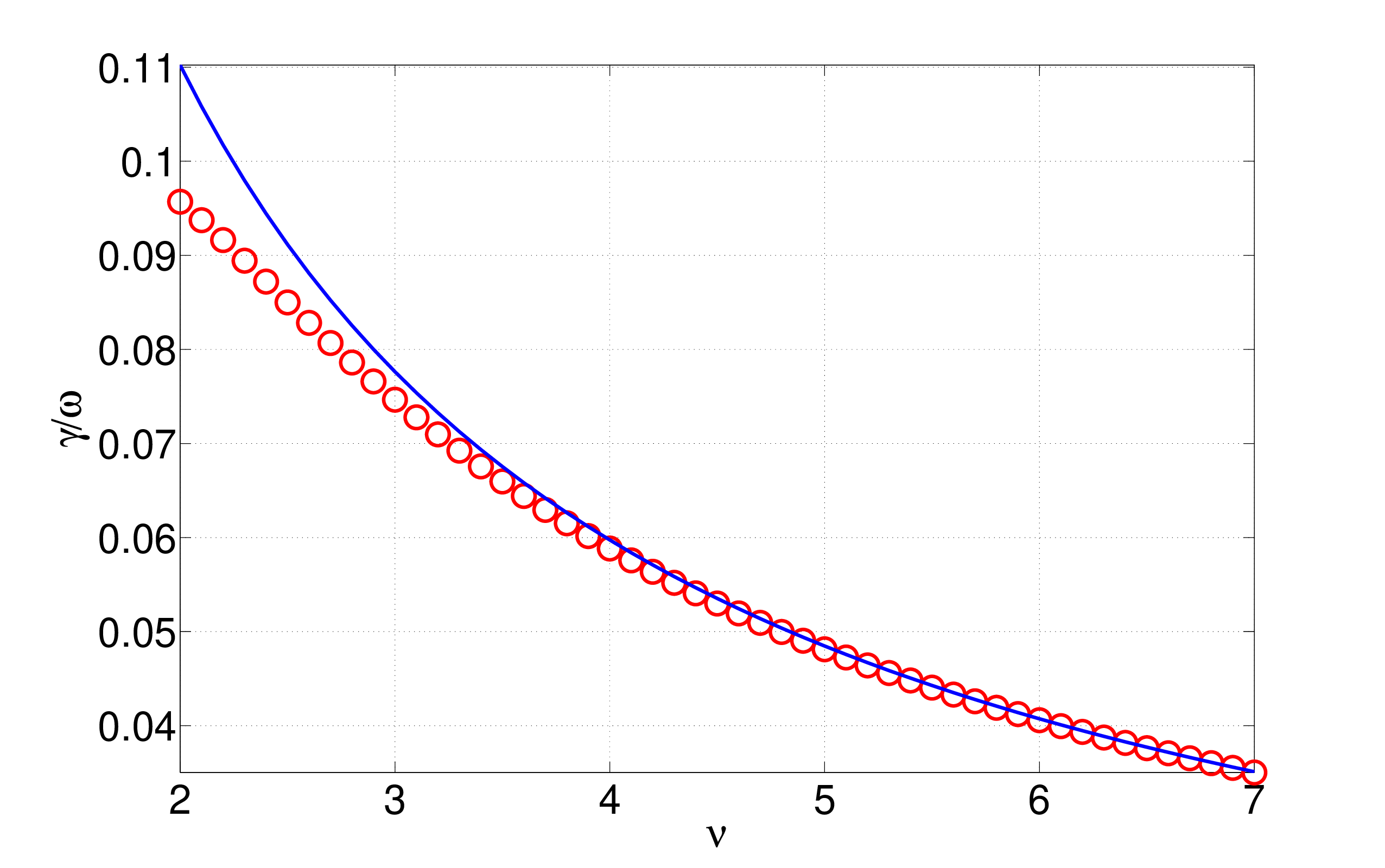}
\end{center}
\caption{Graph of the normalized Lyapunov exponent $\gamma/\omega$ of the Schr\"{o}dinger equation  with white noise potential (solid line) and its asymptotic approximation $\gamma/\omega =0.25\nu^{-1} (1-0.9375\nu^{-2})$ (circles) for large $\dst \nu = {2\omega^3}/{\sigma^2}, \; \omega = \sqrt{|\lambda|},$  and $\lambda > 0$.
} 
\label{pdf2}
\end{figure}

The graphs of the normalized Lyapunov exponent $\gamma/\omega$ \rf{gamma-} as a function of $\dst \nu = {2\omega^3}/{\sigma^2}, \;\; \omega = \sqrt{|\lambda|},$ for $\lambda < 0$ are shown in figures \ref{gamma0_neg} and \ref{gamma1_neg}. For small values of $\nu$ (figure \ref{gamma0_neg})
the relative error of the asymptotic formula \rf{gamma0} becomes less than $0.8\%$ if $\nu < 10^{-3}$. The relative error of \rf{gamma_neg} for large values of $\nu$ (figure \ref{gamma1_neg}) is less than $0.7\%$ as soon as $\nu > 40$. 
\begin{figure}[H]
\begin{center}
\includegraphics[width=0.6\textwidth, angle=0]{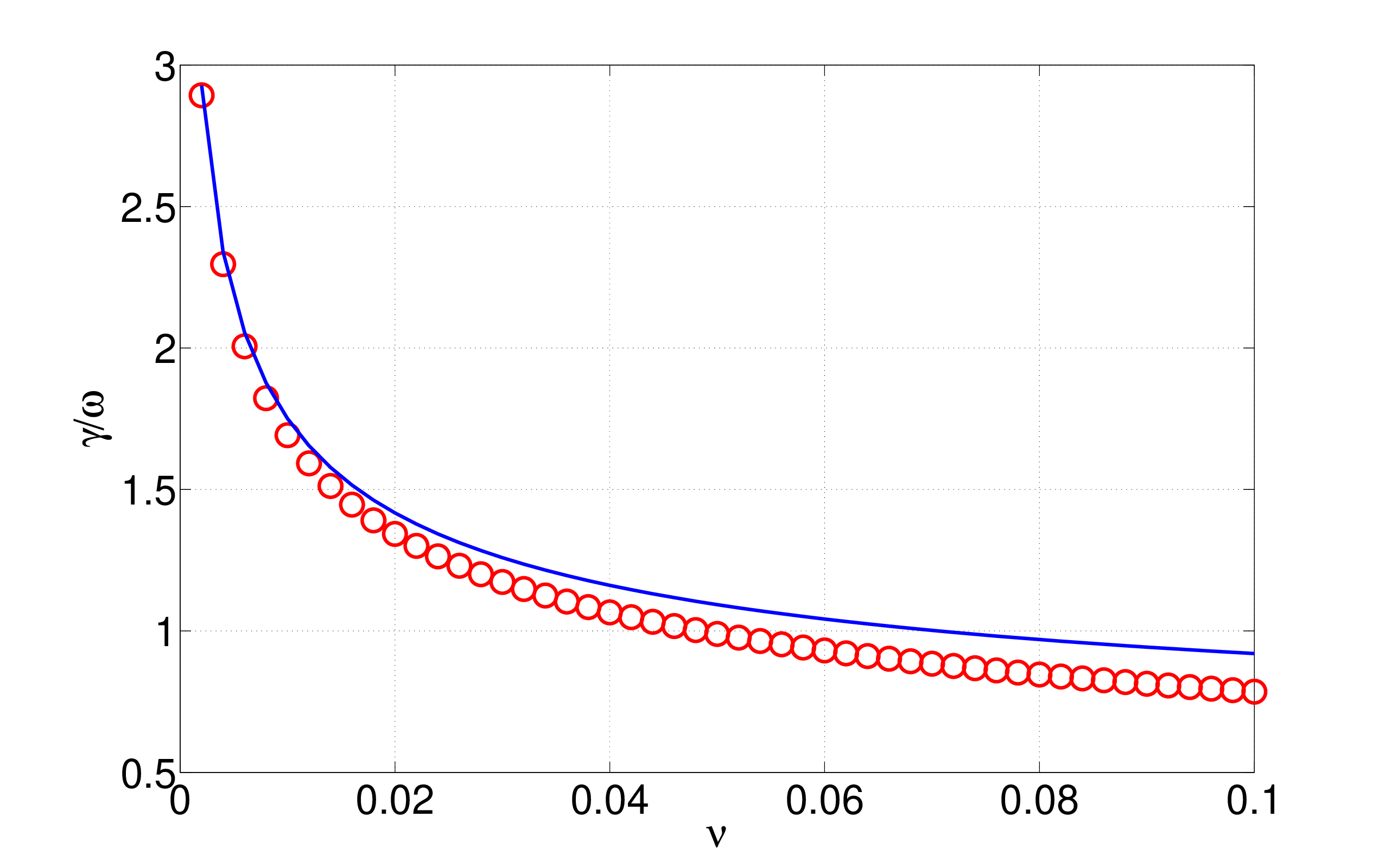}
\end{center}
\caption{Graph of the normalized Lyapunov exponent $\gamma/\omega$ of the Schr\"{o}dinger equation  with white noise potential (solid line) and its asymptotic approximation $\gamma/\omega = 0.3645 \nu^{-1/3}$ (circles) for small $\dst \nu = {2\omega^3}/{\sigma^2}, \; \omega = \sqrt{|\lambda|},$  and $\lambda < 0$.
} 
\label{gamma0_neg}
\end{figure}
\begin{figure}[h!tb]
\begin{center}
\includegraphics[width=0.6\textwidth, angle=0]{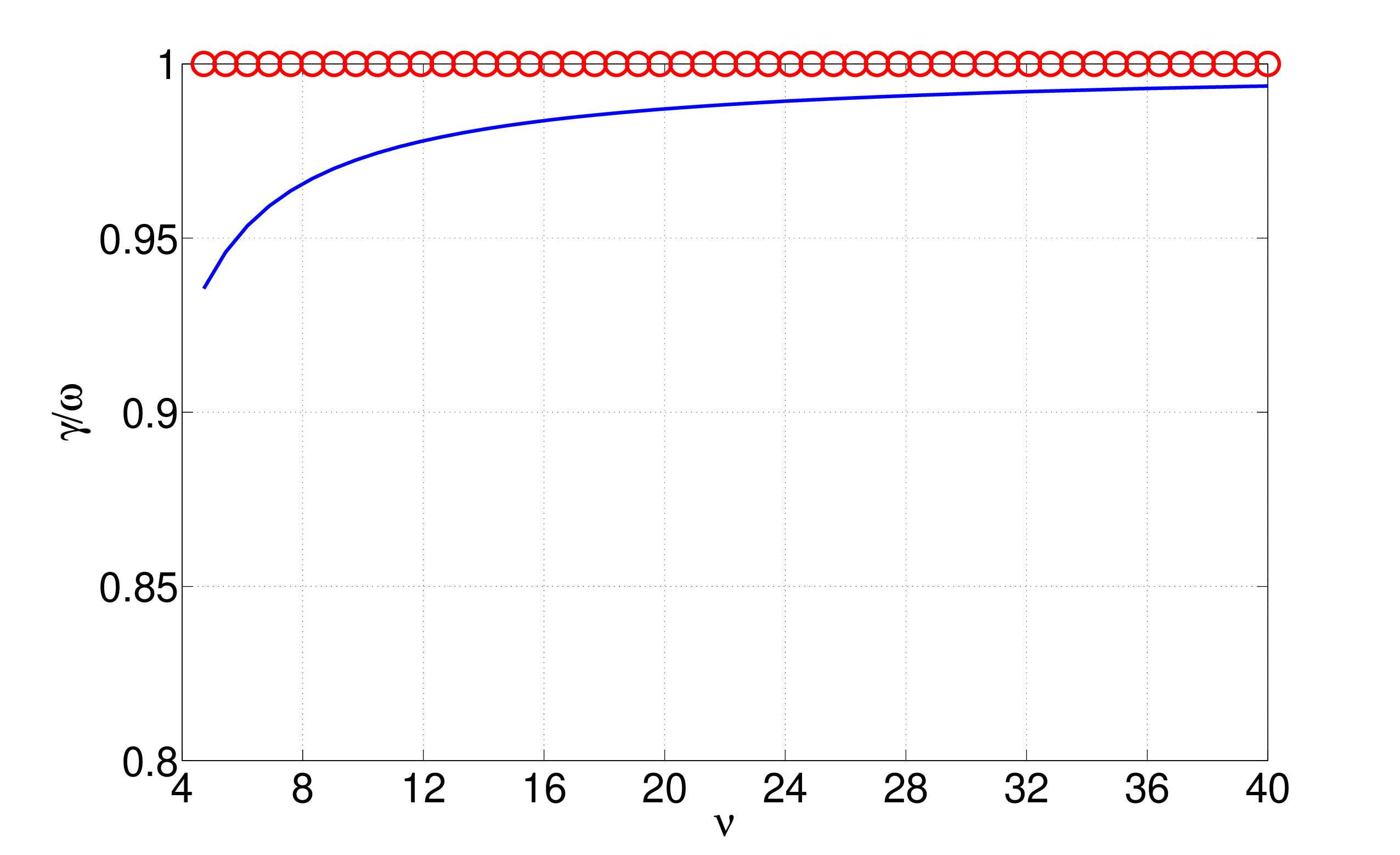}
\end{center}
\caption{Graph of the normalized Lyapunov exponent $\gamma/\omega$ of the Schr\"{o}dinger equation  with white noise potential (solid line) and its asymptotic approximation $\gamma/\omega = 1$ (circles) for large $\dst \nu = {2\omega^3}/{\sigma^2}, \; \omega = \sqrt{|\lambda|},$  and $\lambda < 0$.
} 
\label{gamma1_neg}
\end{figure}

\subsection*{Acknowledgement}

The work of S. Molchanov and B. Vainberg was partially supported by the NSF grant DMS-1008132.

\end{document}